\RequirePackage{fix-cm}
\documentclass[smallextended]{svjour3}       
\smartqed  
\usepackage{amsmath}
\usepackage{amssymb}
\usepackage{color}
\definecolor{red}{RGB}{255,0,0}
\definecolor{blue}{RGB}{0,0,255}
\definecolor{green}{RGB}{0,255,0}

\newcommand {\abs}[1]  {\left\vert#1\right\vert}
\newcommand {\set}[1]  {\left\{#1\right\}}

\newcommand {\defined} {\stackrel{def} {=}}

\newcommand {\nph}     {\textsc{NP}\textrm{-hard}}

\newcommand {\runningtitle}[1] {\vspace{0.5ex}\noindent{\textbf{\boldmath #1:}}}

\newcommand {\commentfig}[1] {#1}

\newtheorem{obs}   {Observation}

\usepackage{amssymb,amsmath}
\usepackage[noblocks]{authblk}
\usepackage{algorithm}
\usepackage{algorithmicx}
\usepackage[noend]{algpseudocode}
\usepackage{algpascal}
\usepackage{graphicx}
\usepackage[usenames,dvipsnames]{pstricks}
\usepackage{graphics}
\usepackage{epsfig}
\usepackage{enumerate}
\usepackage{paralist}
\usepackage{lineno}
\usepackage{amsmath}
\usepackage{mathtools}
\usepackage[hyphens]{url}
\newcommand{\benpg}[1] {\textsc{B}_{{#1}}\textsc{-ENPG}}
\newcommand{\boneenpg} {\benpg{1}}
\newcommand{\maxcut} {\textsc{MaxCut}}
\DeclareMathOperator{\cs}{cs}
\DeclareMathOperator{\swp}{swp}
\DeclareMathOperator{\rot}{rot}
\newcommand{\cutsize}[1] {\cs(#1)}
\newcommand{\swap}[2] {\swp(#1,#2)}
\newcommand{\rotate}[3] {\rot(#1,#2,#3)}
\newcommand{\revrotate}[3] {\rot^{-1}(#1,#2,#3)}
\newcommand{\vect}[1]{\mathbf{#1}}

\newcommand{\cc}{{\cal C}}
\newcommand{\ff}{{\cal F}}
\newcommand{\ct}{{\cal T}}
\newcommand{\tw}[1]{\ct_{#1}}

\begin{document}
\sloppy
\title{The Maximum Cut Problem in Co-bipartite Chain Graphs\thanks{This work is supported in part by TUBITAK Career Project Grant no: 111M482, and TUBITAK 2221 Program.}
}


\author{Arman Boyac{\i} \and
        T{\i}naz Ekim \and \newline
        Mordechai Shalom
}


\institute{A. Boyac{\i} \at
              Department of Industrial Engineering, Bogazi\c{c}i University, Istanbul, Turkey \\
              \email{arman.boyaci@boun.edu.tr}           
           \and
           T. Ekim \at
              Department of Industrial Engineering, Bogazi\c{c}i University, Istanbul, Turkey \\
              \email{tinaz.ekim@boun.edu.tr}
           \and
           M. Shalom \at
	          TelHai College, Upper Galilee, 12210, Israel, cmshalom@telhai.ac.il \\
              Department of Industrial Engineering, Bogazi\c{c}i University, Istanbul, Turkey \\
      	    \email{cmshalom@telhai.ac.il}
}


\maketitle

\begin{abstract}
A \emph{co-bipartite chain} graph is a co-bipartite graph in which the neighborhoods of the vertices in each clique can be linearly ordered with respect to inclusion. It is known that the maximum cut problem ($\maxcut$) is $\nph$ in co-bipartite graphs \cite{Bodlaender00onthe}. We consider $\maxcut$ in co-bipartite chain graphs. We first consider the twin-free case and present an explicit solution. We then show that $\maxcut$ is polynomial time solvable in this graph class.

\keywords{Maximum Cut \and Co-bipartite Graphs \and Dynamic Programming \and Chain Graph}
\subclass{68R10 \and 05C85}
\end{abstract}

\section{Introduction}\label{sec:intro}
A \emph{cut} of a graph $G=(V(G),E(G))$ is a partition of $V(G)$ into two subsets $S, \bar{S}$ where $\bar S=V(G)\setminus S$. The \emph{cut-set} of $(S,\bar{S})$ is the set of edges of $G$ with exactly one endpoint in $S$. The (unweighted) maximum cut problem ($\maxcut$) is to find a cut with a maximum size cut-set, of a given graph. $\maxcut$ has applications in statistical physics and circuit layout design \cite{barahona1988application}.

$\maxcut$ is a widely studied problem; it is in fact one of the 21 $\nph$ problems of Karp \cite{K72}. It is shown that $\maxcut$ remains $\nph$ when restricted to the following graph classes: chordal graphs, undirected path graphs, split graphs, tripartite graphs, co-bipartite graphs \cite{Bodlaender00onthe}, unit disk graphs \cite{DK2007} and total graphs \cite{Guruswami1999217}. On the positive side, it was shown that $\maxcut$ can be solved in polynomial-time in planar graphs \cite{hadlock1975finding}, in line graphs \cite{Guruswami1999217} and
the class of graphs factorable to bounded treewidth graphs \cite{Bodlaender00onthe}. The last class includes cographs and bounded treewidth graphs and we describe it in detail in our work.

A \emph{co-bipartite chain} graph is a co-bipartite graph such that the neighborhoods of the vertices in each clique can be linearly ordered with respect to inclusion. It is first introduced in \cite{heggernes2007linear} and in the same paper the authors present a polynomial-time recognition algorithm.

In our recent work \cite{BESZ14-2} we show that every $\boneenpg$ co-bipartite graph contains at most 4 vertices whose removal leave a co-bipartite chain graph. 

In this work we first consider the twin-free co-bipartite chain graphs and show that these graphs are not factorable to bounded treewidth graphs. Therefore, the algorithm described in \cite{Bodlaender00onthe} is not applicable to these graphs. We present a maximum cut of these graphs having a specific structure. We continue with the general case (allowing twins) and propose a polynomial-time algorithm for $\maxcut$ in co-bipartite chain graphs. We now proceed with definitions and preliminary results.

\runningtitle{Graph notations and terms}
Given a simple graph (no loops or parallel edges) $G=(V(G),E(G))$ and a vertex $v$ of $G$, $uv$ denotes an edge between two vertices $u,v$ of $G$. We also denote by $uv$ the fact that $uv \in E(G)$. We denote by $N(v)$ the set of neighbors of $v$. Two adjacent (resp. non-adjacent) vertices $u,v$ of $G$ are \emph{twins} (resp. \emph{false twins}) if $N_G(u) \setminus \set{v} = N_G(v) \setminus \set{u}$. A vertex having degree zero is termed \emph{isolated}, and a vertex adjacent to all other vertices is termed \emph{universal}. For a graph $G$ and $U \subseteq V(G)$, we denote by $G[U]$ the subgraph of $G$ induced by $U$, and $G \setminus U \defined G[V(G) \setminus U]$. For a singleton $X = \set{x}$ and a set $Y$, $Y + x \defined Y \cup \set{x}$ and $Y - x \defined Y \setminus \set{x}$. A vertex set $U \subseteq V(G)$ is a \emph{clique} (resp. \emph{stable set}) (of $G$) if every pair of vertices in $U$ is adjacent (resp. non-adjacent). An \emph{automorphism} on a graph $G$ is a permutation $\pi$ of $V(G)$, such that $uv$ if and only if $\pi(u)\pi(v)$.

\runningtitle{Some graph classes}
A graph is \emph{bipartite} if its vertex set can be partitioned into two independent sets $V,V'$. We denote such a graph as $B(V,V',E)$ where $E$ is the edge set. A graph $G$ is \emph{co-bipartite} if it is the complement of a bipartite graph, i.e. $V(G)$ can be partitioned into two cliques $K, K'$. We denote such a graph as $C(K,K',E)$ where $E$ is the set of edges that have exactly one endpoint in $K$.

A \emph{bipartite chain graph} is a bipartite graph $G=B(V,V',E)$ where $V$ has a nested neighborhood ordering, i.e. its vertices can be ordered as $v_1,v_2,\ldots$  such that $N_G(v_1) \subseteq N_G(v_2) \subseteq \cdots$. $V$ has a nested neighborhood ordering if and only if $V'$ has one \cite{yannakakis1981node}. Theorem 2.3 of \cite{HPS90} implies that if $G=B(V,V',E)$ is a bipartite chain graph with no isolated vertices, then the number of distinct degrees in $V$ is equal to the number of distinct degrees in $V'$.

A co-bipartite graph $G=C(K,K',E)$ is a \emph{co-bipartite chain} (also known as co-chain) graph if $K$ has a nested neighborhood ordering \cite{heggernes2007linear}. Since $K \subseteq N_G(v)$ for every $v \in K$, the result for chain graphs implies that $K$ has a nested neighborhood ordering if and only if $K'$ has such an ordering.

\runningtitle{Cuts}
We denote a cut of a graph $G$ by one of the subsets of the partition. $E(S,\bar{S})$ denotes the \emph{cut-set} of $S$, i.e. the set of the edges of $G$ with exactly one endpoint in $S$, and $\cutsize{S} \defined \abs{E(S,\bar{S})}$ is termed the \emph{cut size} of $S$. A maximum cut of $G$ is one having the biggest cut size among all cuts of $G$. We refer to this size as the \emph{maximum cut size} of $G$. Clearly, $S$ and $\bar{S}$ are dual; we thus can replace $S$ by $\bar{S}$ and $\bar{S}$ by $S$ everywhere. In particular, $E(S,\bar{S})=E(\bar{S},S)$, and $\cutsize{S}=\cutsize{\bar{S}}$. For an automorphism of $G$, $\pi(S)$ is the cut $S'$ such that $v \in S'$ if and only if $\pi(v) \in S$. In other words $S$ and $\pi(S)$ are identical up to automorphism. In particular, $\cutsize{\pi(S)}=\cutsize{S}$. Let $R$ be a random set of vertices where each vertex of $V(G)$ is chosen to $R$ with probability $1/2$, independent of other vertices. It is easy to see that the expected value of $\cutsize{R}$ is $\abs{E(G)}/2$. Therefore, the size of a maximum cut of a graph is at least $\abs{E(G)}/2$. 

\section{Twin-free Co-bipartite Chain Graphs}\label{sec:twinfree}
\subsection{The structure of twin-free co-bipartite chain graphs}
Let $G=C(K,K',E)$ be a twin-free co-chain graph, and let $G_B=B(K,K',E)$ be the corresponding bipartite (chain) graph. Suppose that $G_B$ does not have isolated vertices and let $k$ be the number of distinct degrees of the vertices of $K$, which by Theorem 2.3 of \cite{HPS90} is equal to the number of distinct degrees of the vertices of $K'$. By the nested neighbourhood property, two vertices of $K$ (resp. $K'$) with the same degree are false twins in $G_B$, thus twins in $G$. Since $G$ is twin-free, $K$ (resp. $K'$) consists of exactly $k$ vertices with distinct non-zero degrees. Every such degree is between 1 and $\abs{K'}=k$. By the pigeonhole principle, $K$ (resp. $K'$) has exactly one vertex of each possible degree. Let $u$ (resp. $u'$) be the unique vertex of $K$ (resp. $K'$) with degree $k$ in $G_B$. We observe that $u$ (resp. $u'$) is adjacent to all the vertices of $K'$ (in both $G$ and $G_B$), and that it is also adjacent to all the vertices of $K$ (in $G$). Therefore, $u$ and $u'$ are universal and twins in $G$, contradicting our assumption. We conclude that at least one of $K$ and $K'$ contains a vertex isolated in $G_B$. We assume without loss of generality that $K$ contains such a vertex. Then, for a given $k$, there are two twin-free co-bipartite chain graphs depending on whether $G_B$ has an isolated vertex in $K'$. We denote by $CC_k$ (resp. $CC_k^-$) the graph containing (resp. lacking) such a vertex. Figure \ref{fig:CC6} depicts the graph $CC_k$. Therefore, the class of twin-free co-bipartite chain graphs is $\set{CC_k, CC_k^- | k \in \mathbb{N}}$. Note that $k$ is the cardinality of $K$ and $K'$ excluding the isolated vertices of $G_B$.

Let $v_i$ be the vertex of $K$ having degree $i$ in $G_B$ and $v'_i$ be the vertex of $K'$ having degree $k-i$ in $G_B$. Then $v_iv'_j$ if and only if $j < i$. We term the pair $(v_i,v'_i)$ as the \emph{row} $i$ of $CC_k$ (or $CC_k^-$).
\begin{obs}
~
\begin{enumerate}[i)]
\item{} The permutation $\pi_k$ defined as $\pi_k(v_i) = v'_{k-i}$, $\pi_k(v'_i) = v_{k-i}$ is an automorphism on $CC_k$.
\item {} $CC_k^- = CC_k - v'_k$.
\item {} $v_k$ is universal in $CC_k^-$.
\item {} $CC_{k-1} = CC_k^- - v_k$.
\item {} The permutation $\pi_k^-$ defined as
\[
\pi_k^-(x)=\left\{
\begin{array}{ll}
x & \textrm{if~} x = v_k\\
\pi_{k-1}(x) & \textrm{otherwise}.
\end{array}
\right.
\]
is an automorphism on $CC_k^-$.
\end{enumerate}
\end{obs}
The last observation follows from the previous two.

\newcommand{\monoblock}[1]{[#1]}
\newcommand{\monos}{\monoblock{S}}
\newcommand{\monosbar}{\monoblock{\bar{S}}}
\newcommand{\biblock}[2]{[#1-#2]}
\newcommand{\bis}{\biblock{S}{\bar{S}}}
\newcommand{\bisbar}{\biblock{\bar{S}}{S}}

\begin{figure}
\centering
\commentfig{
\scalebox{1} 
{
\begin{pspicture}(0,-1.4)(10.16,1.4)
\psline[linewidth=0.02cm](5.0978127,0.6)(6.0978127,1.0)
\psline[linewidth=0.02cm](5.0978127,0.2)(6.0978127,1.0)
\psline[linewidth=0.02cm](5.0978127,-0.2)(6.0978127,1.0)
\psline[linewidth=0.02cm](5.0978127,-0.6)(6.0978127,1.0)
\psline[linewidth=0.02cm](5.0978127,-1.0)(6.0978127,1.0)
\psline[linewidth=0.02cm](5.0978127,0.2)(6.0978127,0.6)
\psline[linewidth=0.02cm](5.0978127,-0.2)(6.0978127,0.6)
\psline[linewidth=0.02cm](5.0978127,-0.6)(6.0978127,0.6)
\psline[linewidth=0.02cm](5.0978127,-1.0)(6.0978127,0.6)
\psline[linewidth=0.02cm](5.0978127,-0.2)(6.0978127,0.2)
\psline[linewidth=0.02cm](5.0978127,-0.6)(6.0978127,0.2)
\psline[linewidth=0.02cm](5.0978127,-1.0)(6.0978127,0.2)
\psline[linewidth=0.02cm](5.0978127,-0.6)(6.0978127,-0.2)
\psline[linewidth=0.02cm](5.0978127,-1.0)(6.0978127,-0.2)
\psline[linewidth=0.02cm](5.0978127,-1.0)(6.0978127,-0.6)
\psdots[dotsize=0.12](1.4978125,1.0)
\psdots[dotsize=0.12](1.4978125,0.6)
\psdots[dotsize=0.12](1.4978125,0.2)
\psdots[dotsize=0.12](1.4978125,-0.2)
\psdots[dotsize=0.12](1.4978125,-0.6)
\psdots[dotsize=0.12](1.4978125,-1.0)
\psdots[dotsize=0.12](2.4978125,1.0)
\psdots[dotsize=0.12](2.4978125,0.6)
\psdots[dotsize=0.12](2.4978125,0.2)
\psdots[dotsize=0.12](2.4978125,-0.2)
\psdots[dotsize=0.12](2.4978125,-0.6)
\psdots[dotsize=0.12](2.4978125,-1.0)
\psline[linewidth=0.02cm](1.4978125,0.6)(2.4978125,1.0)
\psline[linewidth=0.02cm](1.4978125,0.2)(2.4978125,1.0)
\psline[linewidth=0.02cm](1.4978125,-0.2)(2.4978125,1.0)
\psline[linewidth=0.02cm](1.4978125,-0.6)(2.4978125,1.0)
\psline[linewidth=0.02cm](1.4978125,-1.0)(2.4978125,1.0)
\psline[linewidth=0.02cm](1.4978125,0.2)(2.4978125,0.6)
\psline[linewidth=0.02cm](1.4978125,-0.2)(2.4978125,0.6)
\psline[linewidth=0.02cm](1.4978125,-0.6)(2.4978125,0.6)
\psline[linewidth=0.02cm](1.4978125,-1.0)(2.4978125,0.6)
\psline[linewidth=0.02cm](1.4978125,-0.2)(2.4978125,0.2)
\psline[linewidth=0.02cm](1.4978125,-0.6)(2.4978125,0.2)
\psline[linewidth=0.02cm](1.4978125,-1.0)(2.4978125,0.2)
\psline[linewidth=0.02cm](1.4978125,-0.6)(2.4978125,-0.2)
\psline[linewidth=0.02cm](1.4978125,-1.0)(2.4978125,-0.2)
\psline[linewidth=0.02cm](1.4978125,-1.0)(2.4978125,-0.6)
\psellipse[linewidth=0.02,dimen=outer](1.4978125,0.0)(0.2,1.4)
\psellipse[linewidth=0.02,dimen=outer](2.4978125,0.0)(0.2,1.4)
\psframe[linewidth=0.02,linestyle=dashed,dash=0.16cm 0.16cm,dimen=outer](3.0978124,1.2)(0.8978125,0.8)
\usefont{T1}{ptm}{m}{n}
\rput(0.4103125,1.01){row 0}
\psdots[dotsize=0.2,fillstyle=solid,dotstyle=square](5.0978127,1.0)
\psdots[dotsize=0.2,fillstyle=solid,dotstyle=square](5.0978127,0.6)
\psdots[dotsize=0.2,dotstyle=square*](5.0978127,0.2)
\psdots[dotsize=0.2,dotstyle=square*](5.0978127,-0.2)
\psdots[dotsize=0.2,dotstyle=square*](5.0978127,-0.6)
\psdots[dotsize=0.2,dotstyle=square*](5.0978127,-1.0)
\psdots[dotsize=0.2,fillstyle=solid,dotstyle=square](6.0978127,1.0)
\psdots[dotsize=0.2,fillstyle=solid,dotstyle=square](6.0978127,0.6)
\psdots[dotsize=0.2,fillstyle=solid,dotstyle=square](6.0978127,0.2)
\psdots[dotsize=0.2,fillstyle=solid,dotstyle=square](6.0978127,-0.2)
\psdots[dotsize=0.2,dotstyle=square*](6.0978127,-0.6)
\psdots[dotsize=0.2,dotstyle=square*](6.0978127,-1.0)
\psellipse[linewidth=0.02,dimen=outer](5.0978127,0.0)(0.2,1.4)
\psellipse[linewidth=0.02,dimen=outer](6.0978127,0.0)(0.2,1.4)
\psframe[linewidth=0.02,linestyle=dashed,dash=0.16cm 0.16cm,dimen=outer](6.6978126,0.36)(4.4978123,-0.34)
\usefont{T1}{ptm}{m}{n}
\rput(8.1625,0.19){a bichromatic block }
\usefont{T1}{ptm}{m}{n}
\rput(7.923281,-0.15){of type $\bisbar$}
\psframe[linewidth=0.02,linestyle=dashed,dash=0.16cm 0.16cm,dimen=outer](6.6978126,-0.44)(4.4978123,-1.16)
\usefont{T1}{ptm}{m}{n}
\rput(8.4225,-0.61){a monochromatic block }
\usefont{T1}{ptm}{m}{n}
\rput(7.6632814,-0.95){of type $\monosbar$}
\end{pspicture} 
}}
\caption{The graph $CC_5$ and a maximum cut of it.}
\label{fig:CC6}
\end{figure}
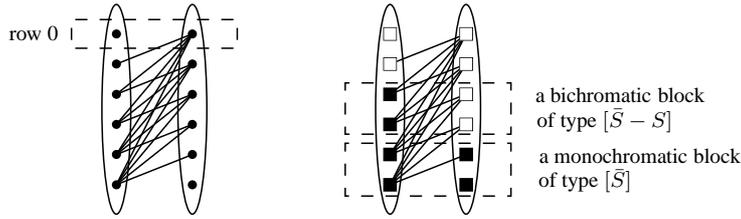

We partition the edges of $G \in \set{CC_k, CC_k^-}$ as: a) \emph{clique edges} of $K$ (resp. $K'$), i.e. edges with both endpoints in $K$ (resp. $K'$) b) \emph{diagonal edges}, i.e. edges $v_iv'_j$ with $j < i$. We note that there are $k^2/2+O(k)$ clique edges of $K$ (resp. $K'$) and $k^2/2+O(k)$ diagonal edges. Therefore, $\abs{E(G)}=3k^2/2+O(k)$, and the size of a maximum cut is at least $3k^2/4+O(k)$. On the other hand, any cut may contain at most half of the clique edges and this is achieved when the cut partitions each clique into two sets of equal size. Therefore, every cut contains at most $k^2+O(k)$ edges. We will show that the maximum cut size of $G$ is $5k^2/6+O(k)$.

\subsection{Inapplicability of known algorithms}

In this section we show that the known algorithms for line graphs \cite{Guruswami1999217} and the class of graphs factorable to bounded treewidth graphs \cite{Bodlaender00onthe} are not applicable to cobipartite chain graphs.

Given a graph $G$, its line graph $L(G)$ is a graph such that each vertex of $L(G)$ represents an edge of $G$, and two vertices of $L(G)$ are adjacent if and only if their corresponding edges share a common endpoint in $G$. It is shown in \cite{Beineke1970129} that $K_5 - e$, the graph obtained by the removal of an edge of $K_5$, is a forbidden subgraph of line graphs. It is easy to verify that $CC_k$ contains $K_5 - e$ as an induced subgraph whenever $k \geq 3$. Therefore, cobipartite graphs are not included in the class of line graphs.

In \cite{Bodlaender00onthe}, a polynomial-time algorithm for $\maxcut$ is presented for a class of graphs that extends both the class of cographs and the class of bounded treewidth graphs. In this section, we show that twin-free co-bipartite chain graphs are not in this class. Therefore, the algorithm in \cite{Bodlaender00onthe} is not applicable to co-bipartite chain graphs. For completeness, we provide a brief definition of the graph classes under consideration.

Cographs are defined inductively as follows:
\begin{itemize}
\item{} A graph with a single vertex ($K_1$) is a cograph.
\item{} If $G_1$ and $G_2$ are cographs then so are their disjoint union and their complete join, i.e the graph obtained by adding the edges $V(G_1) \times V(G_2)$ to the disjoint union.
\end{itemize}
For a cograph $G$, the above recursive definition implies a tree $C(G)$ termed the \emph{cotree} of $G$, in which every leaf corresponds to a vertex of $G$ and the root of $C(G)$ corresponds to $G$.

This definition is extended as follows: Given a graph $H$ with $r>1$ vertices and graphs $H_1,\ldots,H_r$, the graph $G=H[H_1,\ldots,H_r]$ is built by first taking the disjoint union $H_1 \cup \cdots \cup H_r$ and then adding all the possible edges $V(H_i) \times V(H_j)$ for every edge $ij$ of $H$. The collection of graphs $H,H_1,\cdots,H_r$ is termed a \emph{factorization} of $G$. Every graph $G$ has a trivial factorization $G=G[K_1,\ldots,K_1]$. We are typically interested in factorizations such that $H$ is as small as possible. Given a factorization of $G$, one can recursively factorize each of $H_1,\ldots,H_r$. This recursive definition implies a \emph{factor tree} $F(G)$ for $G$. A leaf of $F(G)$ corresponds to a vertex of $G$ and the root of $F(G)$ corresponds to $G$. For any non-leaf vertex $v$ of the factor tree, the graph $H$ used in the factorization is termed the \emph{label graph} of $v$. For a graph class $\cc$, we denote by $\ff(\cc)$ the class of graphs that have a factor tree with label graphs taken from $\cc$. Then, a cotree is a factor tree where a non-leaf vertex is labeled with either $K_2$ or $\bar{K_2}$, i.e. the class of cographs is $\ff(\set{K_2, \bar{K_2}})$.

For any integer $m$, let $\tw{m}$ be the class of graphs having treewidth at most $m$. It is well known that $K_{m+2}$ has treewitdh $m+1$, thus $K_{m+2} \notin \tw{m}$.\footnote{This is the only fact about treewidth used in this work.}

For any integer $m \geq 1$, $\ff(\tw{m})$ is the class of graphs having a factor tree with label graphs of treewidth at most $m$. In \cite{Bodlaender00onthe}, a polynomial-time algorithm for the class $\ff(\tw{m})$ is provided for every constant $m>0$. We now show that the twin-free co-bipartite chain graphs are not contained in this class.

\begin{theorem}
$CC_{m+2} \notin \ff(\tw{m})$ for every $m>0$.
\end{theorem}
\begin{proof}
Consider the root of a factor tree of $CC_{m+2}=C(K,K',E)$. Let $CC_{m+2}=H[H_1,\ldots,H_r]$, and $V(H)=\set{h_1,\ldots,h_r}$ with $r>1$. We will show that the vertices $v_i \in K$ are in distinct graphs among $H_1,\ldots,H_r$. Assume by contradiction that there exist two distinct vertices $v_i, v_j \in K \cap V(H_1)$ such that $i<j$. If some $h_\ell$ is adjacent to $h_1$ in $H$ then $V(H_\ell) \cap K' \subseteq \set{v'_0,\ldots,v'_{i-1}}$. This is because every vertex of $H_\ell$ is adjacent to $v_i$. If some $h_\ell$ is non-adjacent to $h_1$ in $H_0$ then $V(H_\ell) \cap K' \subseteq \set{v'_j,\ldots,v'_{m+2}}$. This is because every vertex of $H_\ell$ is non-adjacent to $v_j$. Therefore,
\[
\left( \cup_{\ell=2}^k V(H_\ell) \right) \cap K' \subseteq \set{v'_0,\ldots,v'_{i-1}} \cup \set{v'_j,\ldots,v'_{m+2}}.
\]
Then $\set{v'_i,\ldots,v'_{j-1}} \subseteq V(H_1)$ and in particular, $v'_i \in V(H_1)$. We now use this fact to prove a stronger property.  Since every vertex of $CC_{m+2}$ is adjacent to either $v_i$ or $v'_i$, every vertex $h_\ell$ of $H$ is adjacent to $h_1$, i.e. $h_1$ is universal in $H$. Therefore, the set of non-neighbours of $h_1$ in $H$ considered in our last argument is empty. Then $V(H_\ell) \cap K' \subseteq \set{v'_0,\ldots,v'_{i-1}}$ for every $\ell \neq 1$ and we conclude that $\set{v'_i,\ldots,v'_j,\ldots,v'_{m+2}} \subseteq V(H_1)$. Since $H_1$ contains two distinct vertices of $K'$, by symmetry we get $\set{v_0,\ldots,v_j} \subseteq V(H_1)$.

Suppose that $j < m+2$, and let $n \in [j+1,m+2]$. Let also $H_n$ be the unique graph among $H_1,\ldots,H_k$ containing $v_n$. Since $h_1$ is universal, $h_n$ is adjacent to $h_1$ in $H$. Then, $v_n$ is adjacent to $v'_{m+2} \in V(H_1)$, implying that $n > m+2$, a contradiction. Therefore, $j=m+2$ and by symmetry $i=1$. We conclude that $H_1=CC_{m+2}$ implying that $k=1$, contradicting the definition of factorization. Therefore, $H_1$ does not contain two distinct vertices $v_i,v_j \in K$. Since $H_1$ is chosen arbitrarily, this holds for every graph $H_\ell$.

Therefore, there are $m+2$ distinct graphs among $H_1,\ldots,H_r$ each of which contains a vertex $v_i \in K$. Since these vertices are pairwise adjacent in $CC_{m+2}$, the $m+2$ vertices corresponding to these graphs are pairwise adjacent in $H$, i.e. they constitute a clique of size $m+2$ in $H$. Therefore, $H \notin \tw{m}$. Since $H$ is chosen as the label of the root of an arbitrary factor tree of $CC_{m+2}$, we conclude that $CC_{m+2} \notin \ff(\tw{m})$.
\qed
\end{proof}

\subsection{The structure of maximum cuts}
We now analyze the structure of a maximum cut of $G$. Let $k$ be such that $G \in \set{CC_k,CC_k^-}$.
Given a cut $(S,\bar{S})$ of $G$, a given row is exactly in one of the sets $S \times S, \bar{S} \times \bar{S}, S \times \bar{S}, \bar{S} \times S$ that we term \emph{row types}. A row is \emph{monochromatic} if it is of one of the first two types and \emph{bi-chromatic} otherwise.
A \emph{block} is a maximal consecutive sequence of rows of the same type. The type of a block is the type of its rows. We denote a monochromatic block as $\monos$ or $\monosbar$ and a bi-chromatic block as $\bis$ or $\bisbar$. The \emph{length} of a block is the number of its rows.

For a cut $S$, $\swap{S}{i}$ is the cut obtained by exchanging the types of the rows $i$ and $i+1$. Formally, let $A=\set{v_i,v_{i+1},v'_i,v'_{i+1}}$ and $B=K \cup K' \setminus A$. Then a)
$\swap{S}{i} \cap B = S \cap B$, b) $v_i \in \swap{S}{i}$ if and only if $v_{i+1} \in S$, c) $v_{i+1} \in \swap{S}{i}$ if and only if $v_i \in S$, d) $v'_i \in \swap{S}{i}$ if and only if $v'_{i+1} \in S$, and e) $v'_{i+1} \in \swap{S}{i}$ if and only if $v'_i \in S$. We note that $\abs{\swap{S}{i} \cap K}=\abs{S \cap K}$ and $\abs{\swap{S}{i} \cap K'}=\abs{S \cap K'}$, i.e. the swap operation preserves the number of clique edges of $S$. In addition, since $\swap{S}{i} \cap B = S\cap B$, all diagonal edges except possibly $v_{i+1}v'_i$ are preserved as well. Therefore, the effect of the swap operation is exactly its effect on the diagonal edge $v_{i+1}v'_i$. The following lemma summarizes this effect.
\begin{lemma}\label{lem:swap}
Let $S$ be a cut and $i,i+1$ two consecutive rows. Then,
\begin{enumerate}[i)]
\item if both rows are monochromatic or both are bi-chromatic then $\cutsize{\swap{S}{i}} = \cutsize{S}$.
\item Otherwise, (i.e. if one row is monochromatic and the other is bi-chromatic) there is a vertex $x \in \set{v_i,v_{i+1},v'_i,v'_{i+1}}$ that is separated from the other three by the cut $S$. Then
\[
\cutsize{\swap{S}{i}} = \cutsize{S} + \left\{
\begin{array}{ll}
1 & \textrm{if~}x \in \set{v_i,v'_{i+1}}\\
-1 & \textrm{otherwise}.
\end{array}
\right.
\]
\end{enumerate}
\end{lemma}

For a cut $S$, let $\rotate{S}{i}{j}$ be the cut obtained from $S$ by shifting one row down the type of all rows from $i$ to $j-1$ and replacing the type of row $i$ by the type (before the shift) of row $j$. Formally,
\[
\rotate{S}{i}{j} = \swap{\swap{\swap{\swap{S}{j-1}}{j-2}{\cdots}}{i-1}}{i}.
\]
Let also $\revrotate{S}{i}{j}$ be the cut obtained in the opposite way. Formally $\revrotate{S}{i}{j}$ is the unique cut $S'$ such that $S=\rotate{S'}{i}{j}$.

\begin{lemma}\label{lem:BlockStructure}
Let $G \in \set{CC_k, CC_k^-}$. There exists a maximum cut $S$ of $G$ such that
\begin{enumerate}[i)]
\item{} $S$ contains at most one block from each type, and\label{itm:AtMostOneBlock}
\item{} the (at most four) blocks of $S$ follow the pattern $\Pi=(\monos,\bisbar,\monosbar,\bis)$ where some of the blocks may be empty.\label{itm:OneStructure}
\end{enumerate}
\end{lemma}
\begin{proof}
\ref{itm:AtMostOneBlock})
Two rows of the same type are termed \emph{separated} if there is a row of a different type between them.
It is sufficient to show that there is a maximum cut with no separated pair of rows. Let $S$ be a maximum cut that contains the smallest number of separated row pairs. If this number is zero then $S$ is the claimed cut. Otherwise, $S$ contains two separated rows $i$ and $j$ with no other rows of the same type between them. Let $S_1=\revrotate{S}{i}{j-1}$, i.e. $S=\rotate{S_1}{i}{j-1}$ and let $S_2=\rotate{S}{i+1}{j}$. We observe that $S$ is obtained from $S_1$ and $S_2$ is obtained from $S$ by the same set of swap operations. Therefore, the effect of these operations on the sizes of the respective cuts is the same. Then
\[
\cutsize{S}-\cutsize{S_1}=\cutsize{S_2}-\cutsize{S}
\]
implying $\cutsize{S}=(\cutsize{S_1}+\cutsize{S_2})/2$. Since $S$ is a maximum cut, we conclude that $S_1$ and $S_2$ are maximum cuts too. Both cuts contain at least one separated pair less than $S$, contradicting the way $S$ was chosen.

\ref{itm:OneStructure})
Let $S$  be a maximum cut with at most one block of each type. If $S$ contains no monochromatic blocks then it follows the pattern $\Pi$ with empty monochromatic blocks. Therefore, $S$ contains a monochromatic block. Assume that the first monochromatic block is $\monos$ where the opposite case is symmetric. Let $i$ be the first row of $\monos$. All the rows before $i$ are bi-chromatic. Consider a row $k \in [i-1]$ of type $\bar{S} \times S$. Let $S'=\rotate{S}{k}{i-1}$. Since all the rows involved in the swap operation are bi-chromatic, by Lemma \ref{lem:swap} we have $\cutsize{S'}=\cutsize{S}$. Row $i-1$ of $S'$ is of type $\bar{S} \times S$. Let $S''=\swap{S'}{i-1}$. Then, by Lemma \ref{lem:swap}, $\cutsize{S''}=\cutsize{S'}+1=\cutsize{S}+1$ contradicting the maximality of $S$.
Therefore, all the rows before the block $\monos$ are of type $S \times \bar{S}$. Similarly, we show that all the rows after $\monos$ until the next monochromatic block are of type $\bar{S} \times S$, and all the rows after $\monosbar$ are of type $S \times \bar{S}$. Without loss of generality we assume that $v_0 \in S$. Therefore, if there is only one monochromatic block the only possible block pattern is $\Pi_3=(\bis,\monos,\bisbar)$; if there are two monochromatic blocks only the patterns $\Pi=(\monos,\bisbar,\monosbar,\bis)$ and $\Pi'=(\bis,\monos,\bisbar,\monosbar)$ are possible. Clearly, $\Pi_3$ is a special case of $\Pi'$ where the last block is empty. We now show that $\Pi'$ is equivalent to $\Pi$, i.e. for every cut that follows pattern $\Pi'$, there is a cut with the same size, following pattern $\Pi$. Let $S'$ be a cut following pattern $\Pi'$. If $G=CC_k$ (resp. $G=CC_k^-$) then the dual of $\pi_k(S')$ (resp. $\pi_k^-(S')$) follows the pattern $\Pi$.
\qed
\end{proof}

We are now ready to prove the main theorem of this section.
\begin{theorem}\label{thm:TwinFree}
A twin-free co-bipartite chain graph $G \in \set{CC_k, CC_k^-}$ has a maximum cut $S$ with block pattern $(\monos,\bisbar,\monosbar,\bis)$, block lengths $x,y,z,t$ respectively, and $\cutsize{S}=\frac{5}{6}k^2+O(k)$ where
\[
(x,y,z,t) = \left(\frac{k}{3}, \frac{k}{3}, \frac{k}{3}, 0 \right) + \left( \delta_x, \delta_y, \delta_z, 0 \right)
\]
and $\abs{\delta_x}, \abs{\delta_y}, \abs{\delta_z} \in [0,1]$.
\end{theorem}
\begin{proof}
By Lemma \ref{lem:BlockStructure}, there is a maximum cut $S$ following the pattern $(\monos,\bisbar,\monosbar,\bis)$. We first consider the case $G=CC_k$. The number of clique edges of $S$ is $(y+z)(x+t) + (x+y)(z+t)$, the number of intra-block diagonal edges is $y(y-1)/2 + t(t-1)/2$, and the number of inter-block diagonal edges is $xy + yz + xz + zt$. By letting $\cutsize{S} = f(x,y,z,t)$ the problem boils down to solving the following system consisting of a quadratic objective function with a single linear equality constraint.
\begin{equation}
\begin{array}{llcl}
\text{maximize} & f(x,y,z,t) & = & (y+z)(x+t) + (x+y)(z+t) \\
    		 & & & + \frac{y}{2}(y-1) + \frac{t}{2}(t-1) + xy + yz + xz + zt \\
\text{subject to~} & x+y+z+t & = & k+1 \\
    		& x,y,z,t & \in & \mathbb{N}\cup \set{0}. \label{eqn:IQP}
\end{array}
\end{equation}
We relax the integrality constraints of (\ref{eqn:IQP}) and calculate the following optimal (fractional) solution $\vect{v^*}$ of the new system \cite{WolframAlphaCCk}.
\[
\vect{v^*} = (x^*,y^*,z^*,t^*) = \left( \frac{k}{3}+\frac{1}{2},~\frac{k}{3},~\frac{k}{3}+\frac{1}{2},~0 \right).
\]
Let $X=\set{\vect{v}=(x,y,z,t)|x+y+z+t=k+1,z=x,t=0}$. Clearly, $\vect{v^*} \in X$.
In the rest of the proof we round $\vect{v^*}$ solution to an optimal integral solution $\vect{\hat{v}} \in X$. We show the optimality of $\vect{\hat{v}}$ by showing that $f(\vect{v^*})-f(\vect{\hat{v}})<1$. Since $f(\vect{\hat{v}})$ is integral whenever $\vect{\hat{v}}$ is integral, this will imply the optimality of $\vect{\hat{v}}$.

Let $\vect{\delta}=\varepsilon (1, -2, 1, 0)$ for some $\varepsilon \in \mathbb{R}$.
Whenever $\vect{v}=(x,y,z,t) \in X$ we have $\vect{v}+\vect{\delta} \in X$, and
\[
f(\vect{v})=f(x,y,z,t)=3x^2+4yx+\frac{y(y-1)}{2}=h(x,y).
\]
Therefore,
\[
f(\vect{v}) - f(\vect{v}+\vect{\delta}) = h(x,y)-h(x+\varepsilon,y-2\varepsilon)=\varepsilon(2x-2y-1+3\varepsilon).
\]
Since $\vect{v^*} \in X$ and $x^*-y^*=1/2$, by substituting in the above equation we get
\[
f(\vect{v^*}) - f(\vect{v^*}+\vect{\delta}) = 3 \varepsilon^2.
\]
For $\abs{\varepsilon} \leq 1/2$ we have $f(\vect{v^*}) - f(\vect{v^*}+\vect{\delta}) < 1$. Therefore $\vect{\hat{v}}=([x^*],k+1-2[x^*],[x^*],0)$ is an optimal integral solution. The value of the optimum is
\[
\cutsize{S}=f(\vect{\hat{v}})=\lfloor f(\vect{v^*}) \rfloor = \lfloor h(x^*, y^*) \rfloor = \left \lfloor \frac{5}{6} k^2 - \frac{3}{2}k + \frac{3}{4} \right \rfloor.
\]
The rest of the proof proceeds similarly for the case of $G=CC_k^-$ and is given in the appendix.
\qed
\end{proof}

\section{The General Case}\label{sec:general}
\newcommand{\zero}{\vect{0}}
\newcommand{\one}{\vect{1}}
\newcommand{\sumvect}[1]{\sum \vect{#1}}
\newcommand{\cochaindynalg}{\textsc{CoChainDynamicProgramming}}
\newcommand{\ksdynalg}{\textsc{KSDecompDynamicProgramming}}

In this section we consider the general case, i.e. graphs that possibly contain twins. Let $G^T$ be a graph possibly containing twins, and let $G$ be a graph obtained from $G^T$ by contracting every set of twins to a single vertex.  The \emph{instances} of a vertex $v \in G$ denoted by $I(v)$ is the set of twins of $G^T$ contracted to $v$. The \emph{multiplicity} of a vertex $v \in G$ is the number of its instances and denoted by $m(v)$. The graph $G^T$ is uniquely defined (up to isomorphism) by the graph $G$ and the multiplicity function $m:V(G) \rightarrow \mathbb{N}$.

For a cut $S$ of $G^T$ and a vertex $v \in V(G)$, $S(v)$ is the number of instances of $v$ in $S$, and $\bar{S}(v) = m(v) - S(v) $ is the number of instances of $v$ in $\bar{S}$. Clearly, $0 \leq S(v), \bar{S}(v) \leq m(v)$, and the cut $S$ is uniquely defined by the cut function $S:V(G) \rightarrow \mathbb{N}$.

Through this section $G^T=C(K,K',E)$ is a co-bipartite chain graph, and $G \in \set{CC_k, CC_k^-}$ is a twin-free co-bipartite chain graph obtained from it by contracting its twins. We will assume without loss of  generality, that $G = CC_k$, since if $G = CC_k^-$ we can set the multiplicity of $v'_k$ to zero. We denote by $G_i$ the subgraph of $G$ induced by the vertices $\set{v_0, \ldots, v_i} \cup \set{v'_0, \ldots, v'_i}$, and by $G^T_i$ the subgraph of $G_T$ induced by the instances of the vertices of $G_i$. Clearly, $G_i$ is a $CC_i$.
Furthermore, we use the vectors $\vect{m}$, $\vect{m}'$, $\vect{s}$, $\vect{s}'$, $\bar{\vect{s}}$, $\bar{\vect{s}}'$ where $m_i \defined m(v_i)$, $m'_i \defined m(v'_i)$, $s_i \defined S(v_i)$, $s'_i \defined S(v'_i)$ and $\bar{\vect{s}}=\vect{m}-\vect{s}$, $\bar{\vect{s}}'=\vect{m}'-\vect{s}'$. We represent the cut $S$ by the pair of vectors $\vect{s},\vect{s}'$. We denote by $\sumvect{x}$ the sum of the entries of a vector $\vect{x}$.

\begin{theorem}\label{thm:recurrenceRelation}
The maximum cut size of a co-bipartite chain graph $G^T$ given by multiplicity vectors $\vect{m}$ and $\vect{m'}$ is
\begin{equation}
\max \set{F_k(x, x')|~0 \leq x \leq \sumvect{m},~0 \leq x' \leq \sumvect{m'}}, \label{eqn:DynamicProggrammingFinal}
\end{equation}
where $F_i(x, x')$ is given by:
\begin{eqnarray}
& F_{-1}(x,x') = & 0,\label{eqn:RecurrenceBase} \\
& F_i(x, x') = & m'_i \cdot x' + m_i (x+x') \nonumber + \\
& & \max_{\substack{ L_i \leq s_i \leq U_i \\ L'_i \leq s'_i \leq U'_i  }}
\begin{cases} F_{i-1}(x - s_i, x' - s'_i) +
 s'_i \left( \sum_{j=0}^{i-1}m'_j - m_i - 2 x'\right)\\
+ s_i \left( \displaystyle \sum_{j=0}^{i-1}(m_j+m'_j) - 2 (x+x') \right) + (s_i + s'_i)^2
\end{cases}\label{eqn:RecurrenceStep}
\end{eqnarray}
with $L_i = \max \left(0, x - \displaystyle \sum_{j=0}^{i-1} m_j\right)$, $U_i = \min \left(m_i, x \right)$, $L'_i = \max \left(0, x' - \displaystyle \sum_{j=0}^{i-1} m'_j \right)$, and $U'_i = \min \left( m'_i, x' \right)$.
\end{theorem}

\begin{proof}
$F_i(x,x')$ denotes the maximum cut size among all cuts $S$ of $G^T_i$ such that $\sum s = x$ and $\sum s' = x'$, i.e.
\[
F_i(x,x')=\max \set{\cs(S)| S \subseteq V(G^T_i), \sumvect{s} = x, \sumvect{s'} = x'}.
\]
With this definition it is clear that the maximum cut size of $G^T$ is given by (\ref{eqn:DynamicProggrammingFinal}). We now provide a recurrence formula for $F_i(x,x')$.

\begin{figure}
\centering
\commentfig{
\scalebox{1} 
{
\begin{pspicture}(0,-1.71)(6.7196875,1.71)
\psellipse[linewidth=0.02,dimen=outer](2.5878124,-0.01)(0.41,1.7)
\usefont{T1}{ptm}{m}{n}
\rput(0.53765625,-1.38){row $i$}
\pscircle[linewidth=0.02,dimen=outer](2.5978124,1.31){0.2}
\pscircle[linewidth=0.02,dimen=outer](2.5978124,0.71){0.2}
\pscircle[linewidth=0.02,dimen=outer](2.5978124,0.11){0.2}
\pscircle[linewidth=0.02,dimen=outer](2.5978124,-1.29){0.2}
\pswedge[linewidth=0.02,fillstyle=solid,fillcolor=black](2.5978124,-1.29){0.2}{-8.130102}{180.0}
\pscircle[linewidth=0.02,dimen=outer](3.9978125,1.31){0.2}
\pscircle[linewidth=0.02,dimen=outer](3.9978125,0.71){0.2}
\pscircle[linewidth=0.02,dimen=outer](3.9978125,-1.29){0.2}
\pswedge[linewidth=0.02,fillstyle=solid,fillcolor=black](3.9978125,-1.29){0.2}{123.69007}{180.0}
\pscircle[linewidth=0.02,dimen=outer](3.9978125,0.11){0.2}
\psellipse[linewidth=0.02,dimen=outer](3.9878125,-0.01)(0.41,1.7)
\pscircle[linewidth=0.02,dimen=outer](2.5978124,-0.49){0.2}
\pscircle[linewidth=0.02,dimen=outer](3.9978125,-0.49){0.2}
\psframe[linewidth=0.02,linestyle=dashed,dash=0.16cm 0.16cm,dimen=outer](4.5978127,1.71)(1.9978125,-0.89)
\pscircle[linewidth=0.02,dimen=outer](1.5978125,0.51){0.2}
\pswedge[linewidth=0.02,fillstyle=solid,fillcolor=black](1.5978125,0.51){0.2}{-75.96375}{180.0}
\pscircle[linewidth=0.02,dimen=outer](4.9978123,0.51){0.2}
\pswedge[linewidth=0.02,fillstyle=solid,fillcolor=black](4.9978123,0.51){0.2}{29.74488}{180.0}
\usefont{T1}{ptm}{m}{n}
\rput(1.9592187,-1.38){$s_i$}
\usefont{T1}{ptm}{m}{n}
\rput(4.599219,-1.38){$s'_i$}
\usefont{T1}{ptm}{m}{n}
\rput(0.64921874,0.42){$x - s_i $}
\usefont{T1}{ptm}{m}{n}
\rput(5.929219,0.42){$x' - s'_i $}
\psline[linewidth=0.02cm](2.5978124,-1.29)(3.9978125,-0.49)
\psline[linewidth=0.02cm](2.5978124,-1.29)(3.9978125,0.11)
\psline[linewidth=0.02cm](2.5978124,-1.29)(3.9978125,0.71)
\psline[linewidth=0.02cm](2.5978124,-1.29)(3.9978125,1.31)
\end{pspicture} 
}}
\caption{A step of the recurrence.}
\label{fig:dynamicprogramming}
\end{figure}
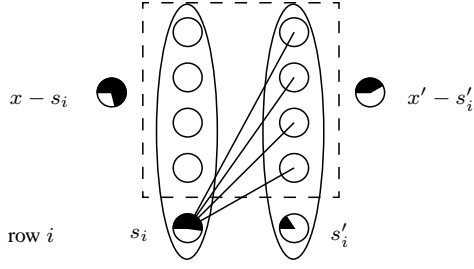

In the base case $CC_{-1}$ is the empty graph, therefore (\ref{eqn:RecurrenceBase}) holds. For the following discussion refer to Figure \ref{fig:dynamicprogramming}. Consider a cut $S$ of the subgraph $G^T_i$ for some $i \geq 0$. $G^T_i$ can be partitioned into a subgraph $G^T_{i-1}$ and two cliques $I(v_i), I(v'_i)$. Therefore the edges of $E(S,\bar{S})$ can be partitioned in the following way according to their endpoints:
\begin{align*}
E_S^1 & = E(S,\bar{S}) \cap E(G^T_{i-1}), \\
E_S^2 & = \left \{uv \in E(S,\bar{S}) | u \in I(v'_i) \right\}, \\
E_S^3 & = \left \{uv \in E(S,\bar{S}) | u \in I(v_i)  \right\}.
\end{align*}
For every instance $u$ of $v_i$ we have $N(u) = V(G^T)-I(v'_i)$ and for every instance $u'$ of $v'_i$ we have $N(u')=K'$. Therefore, the size of these sets are
\begin{align*}
\abs{E_S^1} & = \cs(\underbar{S}) \\
\abs{E_S^2} & = s'_i \left(\sumvect{m'} - x' \right) + \bar{s_i}' \cdot (x'-s'_i) = m'_i \cdot x' + s'_i \left( \sum_{j=0}^{i-1}{m'_j} - 2 x' + s'_i \right)\\
\abs{E_S^3} & = s_i \left( \sumvect{m} - x + \sum_{j=0}^{i-1} m'_j - x' + s'_i \right) + \bar{s_i} \left( x + x' - s'_i - s_i \right)\\
& = m_i (x+x'-s'_i) + s_i \left( \sum_{j=0}^{i-1}(m_j+m'_j) - 2 (x+x'-s'_i) + s_i \right)
\end{align*}
where $\underbar{S}$ is the cut that $S$ induces on $G^T_{i-1}$. Thus $\sumvect{s_-} = x - s_i$ and $\sumvect{s'_-} = x' - s'_i$.
Then $F_i(x,x')$ is the maximum over all possible values of $s_i, s'_i$ of
\[
F_{i-1}(x - s_i, x' - s'_i) + \abs{E_S^2} + \abs{E_S^3}.
\]
As for the possible values of $s_i, s'_i$, recall that $0 \leq s_i \leq m_i$ and $0 \leq s'_i \leq m'_i$. Similarly, $0 \leq x - s_i \leq \sum_{j=0}^{i-1} m_j$ and $0 \leq x' - s'_i \leq \sum_{j=0}^{i-1} m'_j$. Therefore, $F_i(x,x')$ is given by (\ref{eqn:RecurrenceStep}).
\qed
\end{proof}

\begin{theorem}
$\maxcut$ can be solved in time $O(\abs{V(G^T)}^4)$ for a co-bipartite chain graph $G^T$.
\end{theorem}

\begin{proof}
Algorithm \ref{alg:DynAlg} calculates the recurrence relation described in Theorem \ref{thm:recurrenceRelation} through dynamic programming. The running time of function \textsc{CalculateOpt} is proportional to the number of its iterations, i.e. $O(m_i m'_i)$. The running time of the algorithm is proportional to $\sum_{i=0}^k \left(\sum_{j=0}^{i} m_j \right) \left(\sum_{j=0}^{i} m'_j \right) m_i m'_i$. Let $\abs{V(G^T)}=\sumvect{m} + \sumvect{m'}=N$. We proceed as follows
\begin{eqnarray}
& & \sum_{i=0}^k \left(\sum_{j=0}^{i} m_j \right) \left(\sum_{j=0}^{i} m'_j \right) m_i m'_i \nonumber\\
& \leq & \left( \sumvect{m} \right) \left( \sumvect{m'} \right) \sum_{i=0}^k  m_i m'_i
\leq \frac{N^2}{4} \sum_{i=0}^k  m_i m'_i = \frac{N^2}{4} \vect{m} \cdot \vect{m'} \nonumber\\
& \leq & \frac{N^2}{4} \| \vect{m} \|_2 \cdot \| \vect{m'}\|_2 \leq \frac{N^2}{4} \left( \sumvect{m} \right) \left( \sumvect{m'} \right) \leq \frac{N^4}{16}.\nonumber
\end{eqnarray}
\qed
\end{proof}

\alglanguage{pseudocode}

\begin{algorithm}[H]
\caption{$\cochaindynalg$}\label{alg:DynAlg}
\begin{algorithmic}[1]
\Require {$G^T$ is a co-bipartite chain graph}
\Ensure {Return the maximum cut size of $G^T$}
\Statex
\State $G=(K,K',E) \gets$ contract every twin of $G^T$ to a single vertex.
\State \Comment $K$ is a clique with a vertex not adjacent to any vertex of $K'$.
\State $k \gets$ $\abs{K}-1$.
\State $\vect{m} \gets$ the multiplicity vector of the vertices of $K$.
\State $\vect{m'} \gets$ the multiplicity vector of the vertices of $K'$.
\Statex
\State $F(-1,0,0) \gets 0.$
\For {$i = 0$} {$k$}
    \For {$x=0$ \textbf{to} $\sum_{j=0}^i m_i$}
        \For {$x'=0$ \textbf{to} $\sum_{j=0}^i m'_i$}
            \State $F(i,x,x') \gets $\Call{CalculateOpt}{$i,x,x'$}.
        \EndFor
    \EndFor
\EndFor
\State \Return $\max \set{F(k,x,x')|~0 \leq x \leq \sumvect{m}, 0 \leq x' \leq \sumvect{m'}}$.
\Statex

\Function{CalculateOpt}{$i, x, x'$}
\State $w \gets \sum_{j=0}^{i-1} (m_j+m'_j) - 2 (x+x')$
\State $w' \gets \sum_{j=0}^{i-1} m'_j - m_i - 2x'$.
\State $max \gets 0$
\For{$s_i \gets \max(0,x-\sum_{j=0}^{i-1} m_j)$ \textbf{to} $\min (m_i,x)$}
    \For{$s_i' \gets \max(0,x'-\sum_{j=0}^{i-1} m'_j)$ \textbf{to} $\min (m'_i,x')$}
        \State $val \gets F(i-1,x-s_i,x'-s'_i) + w \cdot s_i + w' \cdot s'_i + (s_i + s'_i)^2$.
        \If {$val > max$}
            \State $max \gets val$.
        \EndIf
    \EndFor
\EndFor
\State \Return $max + m_i \cdot (x + x') + m'_i \cdot x'$.
\EndFunction
\end{algorithmic}
\end{algorithm}

We conclude this section with the following remark.
\begin{remark}
The structure of an optimal solution does not necessarily have the structure proven for the twin-free case, namely, three blocks of approximately equal length.
\end{remark}
\begin{proof}
Let $G^T$ be co-bipartite chain graph such that when we contract twins of $G^T$ we end up with $G=CC_9$ and the multiplicity vectors $\vect{m}=(1,1,1,10,1,1,1,1,1)$, $\vect{m'} = (1,1,1,1,1,10,1,1,1)$. Note that $\vect{s}=(1,1,1,3,0,0,0,0,0)$, $\vect{s'}=(1,1,1,1,1,7,0,0,0)$ is a cut as described above yielding a cut-set of size 210. However the output of Algorithm \ref{alg:DynAlg} for the same graph is 223 with the following cut $\vect{s}=(0,0,0,8,1,1,0,0,0)$, $\vect{s'}=(0,0,0,0,0,10,0,0,0)$.
\end{proof}

\section{Summary and Future Work}\label{sec:summary}
In this work, we studied the $\maxcut$ problem in co-bipartite chain graphs. We showed that even the twin-free case cannot be solved using known results, and identified an optimal solution for this case. For the general case we presented a dynamic programming algorithm that constitutes an evidence that the problem is polynomial-time solvable for this graph class.

Finding more efficient algorithms for this class of graphs, and extensions of this technique to other subclasses of co-bipartite graphs is work in progress. The complexity of the weighted $\maxcut$ problem in which edges have associated weights and one has to find a maximum weight cut is unknown for cobipartite chain graphs. The dynamic programming technique does not seem to be applicable to this variant of the problem. On the other hand, for co-bipartite chain graphs one can think about the variant of the problem in which the input is given in compact form, e.g. as two multiplicity vectors. In this case, our dynamic programming algorithm is pseudo-polynomial, since its time complexity is polynomial in the values in the vectors, but not in the size of their binary representations. The complexity of this variant is also an open problem.

\bibliographystyle{abbrvwithurl}
\bibliography{GraphTheory,Optical,Mordo,Approximation,References}

\newpage
\appendix
\huge{Appendix}
\normalsize
\section{Proof of Theorem \ref{thm:TwinFree} continued}\label{sec:appPrelim}
\begin{proof}
Let us consider two maximum cuts $S_1$ and $S_2$ in which $v_k$ is in $S$ and in $\bar{S}$ respectively.
We have $\cutsize{S_1} = g_1(x,y,z,t) = f(x,y,z,t) + 2z + y + t$ and $\cutsize{S_2} = g_2(x,y,z,t) = f(x,y,z,t) + 2x + y + t$ and $x+y+z+t = k$ in both cases. We relax the integrality constraints and calculate the following optimal (fractional) solutions respectively \cite{WolframAlphaCCk-1,WolframAlphaCCk-2}:
\[
\vect{v_1^*} = (x^*,y^*,z^*,t^*) = \left( \frac{k}{3}-\frac{1}{6},~\frac{k}{3} - \frac{1}{3},~\frac{k}{3}+\frac{1}{2},~0 \right),
\]
\[
\vect{v_2^*} = (x^*,y^*,z^*,t^*) = \left( \frac{k}{3}+\frac{1}{2},~\frac{k}{3} - \frac{1}{3},~\frac{k}{3}-\frac{1}{6},~0 \right).
\]
We observe that the values of the two solutions are equal, i.e. $g_1(\vect{v_1^*}) = g_2(\vect{v_2^*}) = \frac{5}{6}k^2 + \frac{5}{6}k + \frac{5}{12}$. In the sequel we round the solution $\vect{v_2^*}$ to obtain an optimal solution of the objective function $g=g_2$.

Let $Y=\set{\vect{v}=(x,y,z,t)|x+y+z+t=k,t=0}$. Clearly, $\vect{v_2^*} \in Y$. For any $\vect{v} \in Y$ we have
\[
g(\vect{v}) = f(\vect{v}) + 2x + y + t = 2x (y+z) + z (x+y) + \frac{y(y+1)}{2}+yz+2x.
\]
Let $\vect{\delta}=(\delta_x , - \delta_x - \delta_z, \delta_z, 0)$ for some $\delta_x, \delta_z \in \mathbb{R}$. Whenever $\vect{v}=(x,y,z,t) \in Y$ we have $\vect{v}+\vect{\delta} \in Y$, and
\[
g(\vect{v}) - g(\vect{v}+\vect{\delta}) = \left( 2x-y-z-\frac{3}{2} \right) \delta_x + \left( 2z-x-y+\frac{1}{2} \right) \delta_z + \left( \delta_x-\delta_z \right)^2 + \delta_x \delta_z.
\]
By substituting $\vect{v}=\vect{v_2^*}$ in the above equation we get,
\[
g(\vect{v_2^*}) - g(\vect{v_2^*}+\vect{\delta}) = -\frac{2}{3} \delta_z + (\delta_x-\delta_z)^2 + \delta_x \delta_z.
\]
We set
\[
(\delta_x, \delta_z) = \left\{
\begin{array}{ll}
(-1/2,1/6) & \textrm{if~}k \equiv 0 \mod 3 \\
(1/6,-1/6) & \textrm{if~}k \equiv 1 \mod 3 \\
(-1/6,1/2) & \textrm{if~}k \equiv 2 \mod 3
\end{array}
\right.
\]
and verify that a) $x+\delta_x$ and $z+\delta_z$ are integral and b) $-\frac{2}{3} \delta_z + (\delta_x-\delta_z)^2 + \delta_x \delta_z < 1$ in each case.
\qed
\end{proof}

\end{document}